\documentclass[10pt, conference, letterpaper]{IEEEtran}
\usepackage{amsmath}
\usepackage{amsthm}
\usepackage{amssymb}
\usepackage{verbatim}
\usepackage{graphicx}

\usepackage{tikz}
\usepackage{pgfplots}
\usetikzlibrary {positioning}
\usetikzlibrary{shapes,snakes}
\usetikzlibrary{arrows,calc}
\usepackage{relsize}
\usetikzlibrary{patterns}
 \usepackage{algorithm}
\usepackage{algpseudocode}

\theoremstyle{plain}
\newtheorem{lemma}{Lemma}
\newtheorem{rem}{Remark}
\newtheorem{theorem}{Theorem}
\newtheorem{assumption}{Assumption}

\newtheorem{corollary}{Corollary}
\theoremstyle{definition}
\newtheorem{example}{Example}
\floatname{algorithm}{Algorithm}

\allowdisplaybreaks[4]

\usepackage{amsfonts}
\usepackage{times}
\usepackage{latexsym}
\usepackage{amssymb}
\usepackage{amsmath}
\usepackage{cite}
\usepackage{verbatim}
\def\bb0{{\mathbb{0}}}

\def\bb{{\mathbf{b}}}

\def\bt{{\mathbf{t}}}

\def\b0{{\mathbf{0}}}

\def\b1{{\mathbf{1}}}

\def\bbE{{\mathbb{E}}}

\def\cA{\mathcal{A}}

\def\cI{\mathcal{I}}

\def\sfM{\mathsf{M}}

\def\sf0{{\mathsf{0}}}

\def\nn{\nonumber}

\newcommand{\marceau}[1]{\color{black}#1\color{black}}

\begin{document}

\newlength{\figurewidth}\setlength{\figurewidth}{0.6\columnwidth}


\addtolength{\topmargin}{-0.5\baselineskip}
\addtolength{\textheight}{\baselineskip}

\title{\fontsize{23}{23}\selectfont Online Knapsack Problem and Budgeted Truthful Bipartite Matching}

\newcounter{one}
\setcounter{one}{1}
\newcounter{two}
\setcounter{two}{2}

\addtolength{\floatsep}{-\baselineskip}
\addtolength{\dblfloatsep}{-\baselineskip}
\addtolength{\textfloatsep}{-\baselineskip}
\addtolength{\dbltextfloatsep}{-\baselineskip}
\addtolength{\abovedisplayskip}{-1ex}
\addtolength{\belowdisplayskip}{-1ex}
\addtolength{\abovedisplayshortskip}{-1ex}
\addtolength{\belowdisplayshortskip}{-1ex}

%
\author{\parbox{3 in}{ \centering Rahul Vaze\\
        School of Technology and Computer Science \\
        Tata Institute of Fundamental Research\\
      	Mumbai, India\\
        {\tt\small vaze@tcs.tifr.res.in}} 
		
}

\maketitle
\begin{abstract}
Two related online problems: knapsack and truthful bipartite matching are considered. For these two problems, the common theme is how to `match' an arriving left vertex in an online fashion with any of the available right vertices, if at all, so as to maximize the sum of the value of the matched edges, subject to satisfying a sum-weight constraint on the matched left vertices.  Assuming that the left vertices arrive in an uniformly random order (secretary model), two almost similar algorithms are proposed for the two problems, that are $2e$ competitive and $24$ competitive, respectively. The proposed online bipartite matching algorithm is also shown to be truthful: there is no incentive for any left vertex to misreport its bid/weight.
Direct applications of these problems include job allocation with load balancing, generalized adwords, crowdsourcing auctions, and matching wireless users to cooperative relays in device-to-device communication enabled cellular network.

\end{abstract}
\vspace{-0.14in}
\section{Introduction} \label{sec:intro}

In this paper, we consider two basic online combinatorial problems : knapsack and truthful bipartite matching,  that have wide applications in practice.
We first consider the {\it online} knapsack problem, where each item, that has two attributes : value and  weight, appears sequentially, and has to be accepted/rejected  irrevocably using only causal information, to maximize the total value of the selected items subject to the sum of their weights being less than the specified capacity. The knapack problem is a classical combinatorial problem, whose online version has also received considerable attention in the literature \cite{marchetti1995stochastic, han2015randomized, babaioff2007knapsack}, as it captures many of the modern resource allocation problems such as generalized adwords, job allocation in cloud computing, load balancing, cognitive radio, admission control and many others \cite{shi2014online, zhang2014dynamic, zheng2013coordinated, cello10generalized, zhang2009resource}. 

The second and related problem to the online knapsack problem is the truthful budgeted bipartite matching problem over a graph $G(L \cup R, E)$, where the right vertex set $R$ is known ahead of time, while left vertices of $L$ arrive sequentially. 
On the arrival of a left vertex $\ell$, utilities of all edges incident on it as well as its bid $c(\ell)$ are revealed. Any left vertex can be matched or accepted only if the payment made to it is larger than $c(\ell)$. With a total payment budget constraint of $C$, the problem is to decide which unmatched vertex of $R$ to match with $\ell$, if at all, immediately and irrevocably, so as to maximize the sum of the utility of all the matched/accepted edges. 
We assume that left vertices are strategic players, which could potentially manipulate the reporting of their true bid, and hence seek a truthful algorithm, i.e., no incoming vertex has any incentive to misreport its bid  to maximize its profit. 

The two problems are closely related, since knapsack problem can be modelled as a bipartite matching problem, 
where all edges incident on a left vertex have same utilities (value of the item) and the capacity constraint on sum-weight is equivalent to the payment budget constraint. Only the truthful aspect is different.

Important applications of the truthful budgeted bipartite matching problem are in crowdsourcing \cite{yang2012crowdsourcing, subramanian2015online} and device-to-device (D2D) cellular wireless communication. The crowdsourcing motivation is exemplified by modern cloud platforms such as Amazon's Mechanical Turk (MTRK),
ClickWorker (CLKWRKR), CrowdFlower (CRDFLWR) that has been well discussed in literature \cite{anari2014mechanism, yang2012crowdsourcing, subramanian2015online, goel2013matching}. In a D2D network, the basic idea is for {\it idle} nodes to help relay other nodes' data to/from the basestation or amongst themselves \cite{asadi2014survey, saad2008distributed}. Since relaying costs resources, each node demands a payment for its help, and the problem is to find an association/matching rule as to who should help whom \cite{gu2015matching} and also the payment to be made for helpers, subject to a total budget constraint on payment. 
To extract largest payment, each node can behave strategically, and hence there is a need for making this association/matching truthful.

To keep both the problems non-degenerate, similar to other prior related works on online algorithms \cite{babaioff2007knapsack, KorulaPal}, we consider a secretarial input model, where the order of arrival of items/left vertices is uniformly random, but their utilities and bids are allowed to be arbitrary. Under this model, we first consider an offline algorithm proposed in \cite{VazeMatching2016} that is useful for both problems, and then use the sample and price  idea to make the algorithms {\it online}. We also make a large market assumption, i.e.,  the utility of any one edge is small compared to the sum-utility of the optimal matching, that is commonly observed in practice for most problems of interest, and is widely used in auction literature \cite{iyer2014mean, gomes2014optimal, mihailescu2010economic}.

To quantify the performance of any online algorithm, we use the well established metric of {\it competitive ratio}, that measures the ratio of the profit of the online algorithm and the optimal offline algorithm (that has access to non-causal information). 
 
 We briefly discuss the prior work on both these problems. The online knapsack problem has been studied widely \cite{marchetti1995stochastic, han2015randomized, babaioff2007knapsack}, with the best known competitive ratio of $10 e$ reported in \cite{babaioff2007knapsack} for a randomized algorithm under the secretarial input. 
The truthful budgeted bipartite matching problem is a special case of a reverse auction \cite{myerson1981optimal}, where users submit bids for accomplishing a set of tasks and if selected, expect a payment at least as much as their reported bids. The offline version of the truthful matching problem, where the full graph is revealed ahead of time, has been considered in \cite{goel2013matching}, where a $3$-approximate algorithm has been derived  that is one-sided truthful. When the goal is to maximize the number of matched edges, \cite{Singer13} provides a $320$-competitive online truthful algorithm assuming the secretarial input model. Under large market assumption, the best known bound for the considered online problem is a  $24 \beta$-competitive  algorithm \cite{VazeMatching2016}, where $\beta$ is the ratio of the largest to the smallest utility of any edge.
Under some additional restrictions such as utilities of all edges incident on a right vertex are identical, a constant-competitive algorithm has been derived in \cite{ZhaoOnlineMatching2014}. 
Our contributions:
\begin{itemize} 
\item

Assuming a large market assumption and secretarial input, we propose a simple algorithm for the online knapsack problem, that is shown to be $2 e$ competitive. Compared to prior work \cite{babaioff2007knapsack}, enforcing the large market assumption, which is mostly satisfied in practice especially in networking problems, we are able to significantly improve the competitive ratio from $10e$ to $2e$. Moreover, our algorithm is also deterministic.

\item The second main contribution of this paper is a $24$-competitive online bipartite matching algorithm that is truthful and satisfies the payment budget constraint.  The previous best known result is a $24 \beta$-competitive algorithm \cite{VazeMatching2016} ($\beta$ is the ratio of the largest to the smallest utility of any edge). Since our algorithm has constant competitiveness, it is scalable and appealing for applications in large networks. 
\end{itemize}

\section{Online Knapsack Problem} Let the value and weight  of item $i \in \cI, |\cI| = n$, be $v(i)$ and $w(i)$, respectively, and the corresponding weight to value ratio (called the buck per bang in the paper) be $b(i) = \frac{w(i)}{v(i)}$. The weights and values (and buck per bang) are arbitrary and allowed to be selected by an adversary. The knapsack problem is to select the set of items that maximizes the sum of their values, subject to a constraint $C$ on the sum of the weight of the items in the selected set.
Thus, without loss of generality, let $w(i) \le C, \ \forall \ i$.

We consider the online knapsack problem, and to keep it non-degenerate in terms of competitive ratio, we assume that the order of arrival of items is uniformly random (secretary-model), i.e., each permutation over $n$ arriving items is equally likely. 
Let $\pi$ be a uniformly random permutation over $[1:n]$. Then the the $k^{th}$ item that 
arrives has value $v(\pi^{-1}(k))$, weight $w(\pi^{-1}(k))$, and buck per bang $b(\pi^{-1}(k))$. 
Under this model, we also assume that given two items arriving at locations $\pi(i)$ and $\pi(j)$, if $b(i) > b(j)$,  then $P(w(i)> w(j)) = \frac{1}{2}$ which is reasonable for most applications.

For a set $S$, we let $v(S)  = \sum_{s \in S} v(s)$.
For any online algorithm $A$ (where on arrival of item $i$, it has to be either accepted/rejected instantaneously and irrevocably), the competitive ratio for solving the knapsack problem is defined as 
$$ \mu_A = \min_{\cI}\frac{\bbE_{\pi}\left\{\sum_{s\in S_A}v(s)\right\}}{v(\mathsf{OPT}(C))},$$ where $\mathsf{OPT}(C)$ is the optimal offline set of selected items and $S_A$ is the set of items selected by $A$, with sum weight constraint $C$. The online knapsack problem is to find the best algorithm $A$ that maximizes the competitive ratio $\mu_A$. $A$ is said to be $\alpha > 1$ competitive if $\mu_A =1/ \alpha$.

We map the knapsack problem to a matching problem,\footnote{The degree of any left or right vertex can be at most $1$.} where we define a bipartite graph $G = (L \cup R, E)$ whose each left vertex $\ell \in L$ corresponds to item $\ell \in \cI$ ($|L| = |\cI|$), and the number of right vertices $|R| = |\cI|$, and edge set $E = \{e = (\ell,r): v(e) = v(\ell), \ \forall  \ r\in R\}$. Thus, each edge incident on left vertex $\ell$ has the same value. Finding the max-weight matching $\sfM$ in $G$ in an online manner, such that $\sum_{e = (\ell,r)\in \sfM} w(\ell) \le C$ is equivalent to solving the online knapsack problem, where on arrival of each left vertex it has to be matched or permanently left unmatched, instantaneously and irrevocably. From hereon, we entirely focus on finding an efficient bipartite online matching subject to capacity constraint $C$.

\begin{assumption}\label{ass:1} Let $v_{max} = \max_{e \in E} v(e)$, and $v(\mathsf{OPT}(C))$ be the optimal value of the matching under the capacity constraint.
We assume the typical large market assumption \cite{goel2013matching}, i.e., $\frac{v_{max}}{v(\mathsf{OPT}(C))} = o(1)$, thus, no single user can influence the outcome significantly.
\end{assumption}
Similar to buck per bang of left vertex, we define for each edge $e = (\ell, r)$ a buck per bang $b(e) = \frac{w(e)}{v(e)}$ that represents the \marceau{weight/cost per unit utility}. For any $\gamma$, let 
$G(\gamma)$ be the graph obtained by removing all edges $e \in E(G)$ with buck per bang $b(e) > \gamma$. Then the proposed online max-weight algorithm \textsc{ON} for solving the online knapsack problem is as given by Algorithm \ref{alg:on}.
\begin{algorithm}
\caption{$\mathsf{ON}$ Algorithm}\label{alg:on}
\begin{algorithmic}[1]
\State {\bf Input:} $L$ set of left vertices/users that arrive sequentially in order $\pi$, $R$ set of right vertices, Capacity $C$ 
\State \%Offline Phase 
\State $L_{t}$ = first $t$ left vertices of $L$
\State Run \textsc{\textsc{Threshold}} on $G_{t}= (L_{t} \cup R, E_{t})$ to obtain $\gamma_{t}\triangleq\gamma_C(G_{t})$ and matching $\sfM_{t}$
\For{each right vertex $r\in R$}
	\If{$e = (\ell,r) \in \sfM_{t}$}
	\State  Set  $\text{price}(r):=b(\ell)$, $\text{cost}(r):=w(\ell)$  
	\Else \State $\text{price}(r):=0$ $\text{cost}(r):=0$
	\EndIf
	\EndFor
	\State \%Decision Phase
\State$\sfM_{\mathsf{ON}} =\emptyset$ 
\State$R' = \{r \in R : \text{price}(r) > 0\}$.
\For{every new left vertex $\ell \in L \backslash L_{t}$},
	\If{$b(\ell) = \frac{w(\ell)}{v(\ell)}>\gamma_{t}$}
	\State \%Pruning:  Let $\ell$ be permanently unmatched 
	\State Break
	\Else	\State Let $e^\star=(\ell, r)$ be the edge with the smallest $\text{price}(r), r\in R'$ such that $b(\ell) < \text{price}(r)$ and $w(\ell) < \text{cost}(r)$
	 	\If{$\sfM_{\mathsf{ON}} \cup  \{e^{\star}\}$ is a matching } 
	\State $\sfM_{\mathsf{ON}} = \sfM_{\mathsf{ON}}  \cup \{e^\star\}$ 
	\Else 
	\State Let $\ell$ be permanently unmatched,
	\EndIf
	\EndIf
\EndFor
\end{algorithmic}
\end{algorithm}
The idea behind $\mathsf{ON}$ is as follows:
\begin{itemize}
\item Do not match any of the first $t$ left vertices (called the offline phase), and only use them to run the offline \textsc{Threshold} algorithm \cite{VazeMatching2016} and find the threshold $\gamma_{t}$ and the matching $\sfM_{t}$ with capacity $C$.
\item For any right vertex $r$, such that $e = ( *, r) \in \sfM_{t}$, set its  $\text{price}(r)$ and  $\text{cost}(r)$ to be the buck per bang and the weight of the left vertex matched to $r$ in $\sfM_{t}$, respectively. 
\item In the decision phase, starting with the arrival of $t+1^{st}$ left vertex, do not consider it for selection if its buck per bang $b(e)$ larger than $\gamma_{t}$. Otherwise, match the newly arrived left vertex $\ell$ to the available/unmatched right vertex $r$ with the smallest price that is larger than the buck per bang $b(\ell)$ of $\ell$ and has weight less than the cost of $r$. Thus the number of selected/matched left vertices is at most the number of left vertices matched by the \textsc{Threshold} algorithm in the offline phase.
\end{itemize}

Before proving results on $\mathsf{ON}$, we first consider the subroutine (\textsc{Threshold} algorithm \cite{VazeMatching2016}) that is used to generate an offline matching with the first $t$ left vertices, where the \textsc{Greedy} subroutine is the usual greedy matching algorithm for a bipartite graph. Essentially, the \textsc{Threshold} algorithm tries to find 
the largest threshold $\gamma_C$ such that the sum-weight of the edges that are part of the greedy matching on the edges with buck per bang less than the threshold, satisfies the capacity constraint. 
 
\begin{algorithm}
\caption{\textsc{Threshold}}\label{alg:unimech}
\begin{algorithmic}[1]
\State {\bf Input:} Graph $G$, Capacity $C$
\State {\bf Output:} Matching $\sfM$, Threshold $\gamma_C$
	\State$\cA(G) = \{\gamma : \sum_{e\in \sfM}\gamma v(e)\leq C,\; \sfM=\mbox{\textsc{Greedy}}(G(\gamma))\}$
	\State $\gamma_C=\max\{\gamma: \gamma \in \cA(G) \}$
\State Accept all users in $\sfM = \mbox{\textsc{Greedy}}(G(\gamma_C))$ 

\end{algorithmic}
\end{algorithm}

\begin{rem}\label{rem:threshold} The matching $\sfM$ output by \textsc{Threshold} algorithm for graph $G$ is a Greedy matching for graph $G(\gamma_C)$. Moreover, since all matched left vertices have $b(e) \le \gamma$, and from the definition of \textsc{Threshold} algorithm, $\gamma_C \sum_{e \in \sfM} v(e) \le C$, we have  $\sum_{\ell: e=(\ell,r) \in \sfM} w(\ell) \le C$, i.e., $\sfM$ satisfies the capacity constraint. 
\end{rem}

We next list some important properties of \textsc{Threshold} algorithm \cite{VazeMatching2016}, whose proofs are  presented in the Appendices for completeness sake.

\begin{lemma}\label{lem:umguarantee} \cite{VazeMatching2016} Let $\sfM(\textsf{off})$ be the matching output by \textsc{Threshold} algorithm with input graph $G$ under capacity constraint $C$. Then under Assumption \ref{ass:1}, $v(\sfM(\textsf{off})) \ge \frac{\mathsf{OPT}(C)}{3+o(1)}$.
\end{lemma}
Lemma \ref{lem:umguarantee} is valid for all graphs, but if we restrict to a special class of graphs considered in this section, where values of all edges incident on any left vertex are identical and the number of right and left vertices are equal, we can get a better bound as a corollary to Lemma \ref{lem:umguarantee} as follows.
\begin{corollary}\label{cor:umguarantee}
Let $\sfM(\textsf{off})$ be the matching output by \textsc{Threshold} algorithm with input graph $G$ (where edge set $E = \{e= (\ell, r): v(e) = v(\ell)\}$) under capacity constraint $C$. Then under Assumption \ref{ass:1}, $v(\sfM(\textsf{off})) \ge \frac{\mathsf{OPT}(C)}{1+o(1)}$.
\end{corollary}

\begin{rem} Assumption \ref{ass:1} is critical in the sense that if it is violated, then the approximation ratio of the \textsc{Threshold} algorithm can be arbitrarily bad which can be showed as follows. Consider the case when there are only two items, $v(1) = 1, w(1) = 1$, and $v(2) = C-1, w(2) = C$, with capacity $C$. The optimal solution is to just choose item $2$ (assumption \ref{ass:1} is not satisfied since $v(\{2\})/v(\mathsf{OPT}(C)) = 1$), while the \textsc{Threshold} algorithm will choose item $1$ and the approximation ratio will be $1/C$. 
\end{rem}

Before analyzing the \textsc{ON} algorithm, we first consider the offline case, when \textsc{Threshold} is run over the full graph $G(L \cup R, E)$ and output threshold is $\gamma$ and matching is $\sfM(\text{off})$.
Recall that the edge weights of all edges incident on any left vertex are identical and the number of left and right vertices are equal. Hence the greedy matching $\sfM(\text{off})$ output by the \textsc{Threshold} 'offline' algorithm (when run on the full graph $G(L \cup R, E)$) contains all the left vertices that have buck per bang less than or equal to the threshold $\gamma$. Let the set of left vertices selected by the \textsc{threshold} algorithm be $L^\star$, i.e., set of left vertices with buck per bang less than $\gamma$. 
From Corollary \ref{cor:umguarantee}, we know that the utility of set $L^\star$ is almost optimal.

In the online case, we now  aim to select as many left vertices of $L^\star$, though without knowing $\gamma$ exactly, since \textsc{Threshold} cannot be run on the full graph $G$. Alternatively, we are trying to select as many left vertices that have buck per bang less than 
$\gamma$. 
This is reminiscent of the $k$-secretary problem, where the objective is to select the $k$ secretaries with the largest utilities in an online fashion. 

 Apart from the major challenge of finding $\gamma$, another minor problem is that  we do not know  how many secretaries we want to pick ahead of time. We overcome both these challenges via algorithm \textsc{ON}, where we first estimate a $\gamma_t \ge \gamma$ by running \textsc{Threshold} on a subgraph $G_t\subseteq G$ (graph consisting of the first $t$ left vertices of $G$), and then select as many left vertices that are matched/selected by running \textsc{Threshold} on graph $G_t$. We show that algorithm \textsc{ON} selects  any left vertex that is part of $L^\star$ with probability at least  $1/2e$.

%


We next state a critical lemma for analyzing the performance of the \textsc{ON} algorithm that shows that the $\gamma_t$ computed in the offline phase of \textsc{ON} is always larger than $\gamma$ (Lemma \ref{lem:monotoneGamma}), and hence all vertices of $G$ that are part of $\sfM(\text{off})$ are not pruned  in Step $17$ of the \textsc{ON} algorithm. 

\begin{lemma}\label{lem:monotonegreedymatching} \cite{VazeMatching2016}
Let $G = (L\cup R, E)$ and $F\subseteq G$, such that $F = (L\backslash L'\cup R, E')$, and the edge set $E'$ is such that all edges incident on left vertices in set $L'$ are removed simultaneously, while all edges incident on $L\backslash L'$ are retained as it is. Then $$v(\textsc{Greedy}(G)) \ge v(\textsc{Greedy}(F)).$$ Moreover $$v(\textsc{Greedy}(G(\gamma_1))) \ge v(\textsc{Greedy}(G(\gamma_2))) \ \text{for} \ \gamma_1 \ge \gamma_2,$$ and $$v(\textsc{Greedy}(G(\gamma))) \ge v(\textsc{Greedy}(F(\gamma))).$$
\end{lemma}
For arbitrary subgraph $F \subseteq  G$ (where any arbitrary edges are removed from $G$), $v(\text{\textsc{Greedy}}(G))$  may or may not be larger than $v(\text{\textsc{Greedy}}(F))$. 
The importance of Lemma \ref{lem:monotonegreedymatching} is in showing that \textsc{Threshold} is solvable in polynomial time and the threshold $\gamma_C$ is monotonic for classes of graphs considered in this paper. In particular, for the bipartite graphs considered in this paper, each left vertex has a fixed weight/cost and the buck-per-bang of edge $e = (\ell, r)$ is $b(e) = \frac{w(\ell)}{v(e)}$. Thus, if any edge $e = (\ell, r)$ has $b(e) > \gamma$, then all edges $e' = (\ell, *)$ for which their value $v(e') <  v(e)$  that are incident on the left vertex $\ell$ also have $b(e) > \gamma$ and are not part of graph $G(\gamma)$.
We prove the two claims as follows.

\begin{lemma}\label{lem:polytimecomplexity} \cite{VazeMatching2016}
\textsc{Threshold} is solvable in polynomial time.
\end{lemma}

Algorithm \textsc{Threshold} involves finding a maximum in Step 4. In the proof, it is shown that bisection can be used to solve this maximization. We would like to note that if $v(\textsc{Greedy}(G(\gamma))) \ngtr v(\textsc{Greedy}(F(\gamma)))$, then finding this maximum is non-trivial.

The following Lemma shows that if \textsc{Threshold} algorithm is run on a (special) subgraph of $G$, then the output threshold $\gamma_C$ increases, which we critically need to show that all left vertices that are part of $L^\star$ are eligible for matching in the \textsc{ON} algorithm.
\begin{lemma}\label{lem:monotoneGamma} \cite{VazeMatching2016}
Let $G = (L\cup R, E)$ and $F = (L\backslash L'\cup R, E')$, where the edge set $E'$ is such that all edges incident on left vertices in set $L'$ are removed simultaneously, while all edges incident on $L\backslash L'$ are retained as it is. Then 
$\gamma_C(F) \ge \gamma_C(G)$.
\end{lemma}

Finally, we are ready to state the first main result of the paper on the expected utility of the online matching $\sfM_{\mathsf{ON}}$, output by the \textsc{On} algorithm.

\begin{theorem}\label{lem:onguarantee} $\bbE\{v(\sfM_{\mathsf{ON}})\} \ge \frac{v(\mathsf{OPT}(C))}{2 e(1+o(1)) }$.
\end{theorem} 
\begin{proof} 
Consider the full graph $G = (L \cup R, E)$ (offline) and its subset $G_t = (L_t \cup R, E_t)$ (offline for \textsc{ON}),  and let $\gamma$ and $\gamma_t$ be the output threshold when  \textsc{Threshold} is run over $G$ and $G_t$, both with capacity $C$, respectively.
From Lemma \ref{lem:monotoneGamma}, it follows that $\gamma_t \ge \gamma$, hence all the left vertices $L^\star$ matched by the \textsc{Threshold} algorithm with the full graph $G$ that arrive in the decision phase are not pruned in Step $17$ with the \textsc{On} algorithm.

 In the decision phase of the  \textsc{On} algorithm, disregard the condition that $w(\ell) < \text{cost}(r)$ for selecting a left vertex for now. Then the 
 left vertex $\ell \in L^\star$ that appears in the decision phase at the $i^{th}$ position, $i> t$, is selected 
as long as it is selected by the \textsc{Virtual} Algorithm \cite{babaioff2007knapsack}. This assertion follows since with the \textsc{Virtual} Algorithm, a left vertex in the decision phase is selected only if its buck per bang is lower than the currently largest price among the right vertices in the reference set $R'$, and more importantly that the current largest price  was derived from the buck per bang of a left vertex that arrived in the offline phase. With algorithm \textsc{On}, a left vertex in the decision phase is selected as long as there is at least one unmatched right vertex with price larger than its buck per bang. Thus, if any left vertex is selected by  \textsc{Virtual} Algorithm then it is definitely selected by the \textsc{On} algorithm. We illustrate the main difference between the \textsc{On} and the \textsc{Virtual} Algorithm via an example as follows.

\begin{algorithm}
\caption{Virtual Algorithm}\label{alg:virtual}
\begin{algorithmic}[1]
\State \%Offline Phase {\bf Input:} $L_t$,  $(\sfM_{t}, \gamma_t) = \textsc{Threshold}(L_t \cup R)$
\State $V = \{r: (\ell, r) \in \sfM_{t}\}$
\State $\text{price}(r) = b(e)$ for $e = (\ell, r) \in \sfM_{t}$
\State Order the elements of $V$ in increasing $\text{price}(r), r\in V$, the element with the largest price is $r_{|V|}$
\State Initialize $S = \Phi$
\State For every new left vertex $\ell\in L\backslash L_t$ in the decision phase
\If{$b(\ell) < \text{price}(r_{|V|})$}

\If  {$r_{|V|}$ was sampled in offline phase}
		\State  $S = S \cup \{\ell\}$
	\EndIf
	\State Update $\text{price}(r_{|V|}) = \gamma(\ell)$
	
\State Order the elements of $V$ in increasing $\text{price}(r), r\in V$
	
	\Else \ Do nothing and keep $\ell$ unmatched
\EndIf
\end{algorithmic}
\end{algorithm}

\begin{example} Consider the input graph $G$, where in the offline phase two left vertices that are matched/selected by the \textsc{Threshold} algorithm are $S = \{s_1, s_2\}$ with 
$\{b(s_1), \  b(s_2)\} = \{ 1/5, \ 1/6\}$. 
Let the left vertices $\ell_1, \ell_2$ (indexed in order of arrival) in the decision phase have 
$\{b(\ell_1), \ b(\ell_2)\} = \{1/5.1,\ 1/7 \}$, respectively.
Then with the \textsc{On} algorithm, on arrival of $\ell_1$ with $b(\ell_1)= 1/5.1$ it is compared with $s_1$ that has $b(s_1)=1/5$ and since $b(\ell_1) < b(s_1)$, $\ell_1$ is selected. Similarly, on arrival of $\ell_2$ with $b(\ell_2)= 1/7$ it is compared with $s_2$ (that has not been compared before and matched) that has $b(s_2)=1/6$, and $\ell_2$ is also selected.
With the \textsc{Virtual} algorithm, the offline matched set $\{b(s_1), \ b(s_2)\} = \{ 1/5, \ 1/6\}$ remains the same as in \textsc{On}. Moreover,  in the decision phase, on arrival of $\ell_1$ with $b(\ell)= 1/5.1$ it is compared with $s_1$ (with worst $b(.)$ value among the two), and since $b(\ell_1) < b(s_1)$, $\ell_1$ is selected. The main difference is in the next step, where the set  $V = \{s_1, s_2\}$  is updated to include $\ell_1$ and eject $s_1$ to get the reference set as $V = \{\ell_1, s_2\}$ with 
$\{b(\ell_1),\ b(s_2)\} =  \{ 1/5.1, \ 1/6\}$. Next, when $\ell_2$ arrives with $b(\ell)= 1/7$, even though it has better buck per bang than both $b(\ell_1)$ and $b(s_2)$, but since the maximum value of $b(.)$ among $\ell_1$ and $s_2$, $1/5.1$ is seen in the decision phase and not in the offline phase; $\ell_2$ is not selected.
\end{example}

From \cite{babaioff2007knapsack}, with the \textsc{Virtual} Algorithm, a new left vertex that appears at location $i$ is selected
 if and only if at location $i$, the left vertex with the largest buck per bang in the virtual set $V$  is sampled at or before time $t$. 
Since the permutations are uniformly random, the probability of this event is $\frac{t}{i-1}$. Hence the probability of selecting $\ell \in L^\star$ when it arrives at position $i \in [t+1, n]$ is 
 \begin{eqnarray} \nn
P(\ell\in L^\star \ \text{is selected}) &=& \sum_{i=t+1}^n \frac{1}{n} \frac{t}{i-1} = \frac{t}{n} \sum_{i=t+1}^n\frac{1}{i-1} \\
&>& \frac{t}{n} \int_t^n \frac{dx}{x} = \frac{t}{n} \ln\left(\frac{n}{t}\right), 
\end{eqnarray}
where the first equality follows since the probability of $\ell$ arriving at the $i^{th}$ location is $\frac{1}{n}$ independent of $i$. Choosing $t = \frac{n}{e}$, maximizes the lower bound, and we get that $P(\ell\in L^\star \ \text{is selected}) = 1/e$. 

Hence by linearity of expectation, we get that the expected value of the selected left vertices by 
\textsc{On} algorithm is at least 
 \begin{equation}\label{eq:exppayoff}
 \bbE\left\{v(\sfM_{\mathsf{ON}})\right\} \ge \sum_{\ell \in L^{\star}} \frac{1}{e} v(\ell)= \frac{1}{e} v(\sfM{\text{off}}).   
 \end{equation}

Now we enforce the condition that $w(\ell) < \text{cost}(r)$ for selecting a left vertex. We show in Lemma \ref{lem:budfeas} that selecting left vertices only when $w(\ell) < \text{cost}(r)$ implies that \textsc{On} algorithm satisfies the sum-weight constraint $C$.
Recall that we have assumed that under the secretarial model of input, given $b(i) > b(j)$, $P(w(i) > w(j)) =\frac{1}{2}$. 
Since each left vertex $\ell$ selected by \textsc{On} algorithm has 
$b(\ell) \le b(j)$ for some left vertex $j$ that is part of offline matching $\sfM_t$. 
Thus, each left vertex that belongs to $\sfM_{\textsf{ON}}$ without enforcing $w(\ell) < \text{cost}(r)$, is selected with probability $1/2$ even when the constraint is enforced, and we get from \eqref{eq:exppayoff}, that 
\begin{equation}\label{eq:m3m1}
\bbE\{v(\sfM_{\textsf{ON}})\}=\frac{1}{2e} v(\sfM{\text{off}}).
\end{equation}
Finally, the result follows since $v(\sfM{\text{off}}) > \frac{v(\mathsf{OPT})}{1+o(1)}$ from Corollary \ref{cor:umguarantee}.
\end{proof}

\begin{lemma}\label{lem:budfeas}
Algorithm $\mathsf{ON}$ satisfies the capacity constraint.
\end{lemma}
\begin{proof}  
Let $\gamma_{t} = \gamma_{C}(G_{t})$ for simplicity. For each $r \in R'$ (right vertices matched in the offline phase), from Remark \ref{rem:threshold} we have that for \textsc{Threshold} algorithm, $\sum_{r\in R'} \text{cost}(r) \le C$. In the decision phase, any left vertex is accepted (is matched to $r\in R'$) if its weight is less than the cost of $r\in R'$, and once $r\in R'$ is matched it is not available thereafter  
(at most $|R'|$ left vertices are selected). Therefore, it directly follows that for the set of matched left vertices in the decision phase $L_D$, $\sum_{\ell\in L_D} w(\ell) \le \sum_{r\in R'} \text{cost}(r)$. Since we know that $\sum_{r\in R'} \text{cost}(r) \le C$, the claim follows.
\end{proof}

{\it Discussion:} In this section, we proposed an online algorithm \textsc{on}  for the knapsack problem with competitive ratio $2 e$, improving upon the currently best known bound of $10 e$ \cite{babaioff2007knapsack}, under an extra large market assumption (Assumption \ref{ass:1}). Assumption \ref{ass:1} is reasonable for most networking applications and has been considered widely in auction  literature \cite{iyer2014mean, gomes2014optimal, mihailescu2010economic}. 
Assumption \ref{ass:1} is also satisfied if the value of items is generated according to a stochastic process that is light-tailed, which is what is generally observed in practice. 
In the next section, we build upon the \textsc{on} algorithm to propose a truthful algorithm for the online bipartite budgeted matching problem.

\section{Truthful Budgeted Bipartite Matching}
Motivated by crowdsourcing and D2D communication applications, in this section, we consider an online matching problem over a bipartite graph $G(L\cup R,E)$, where the right vertex set $R$ is known ahead of time, while left vertices of $L$ arrive sequentially in a random order. The incident edge utilities $v(e), e =(\ell,r), r\in R$ from a vertex $\ell\in L$ to set $R$ are revealed only upon its arrival, as well as its bid $c(\ell)$, and the problem is to decide which unmatched vertex of $R$ to match with $\ell$, if at all, immediately and irrevocably. If vertex $\ell$ is matched, a payment $p_{\ell}$ is made to vertex $\ell$ that has to be at least as much as its reported bid $c(\ell)$.  A total  budget constraint of $C$ is assumed for payments to be made to the matched left vertices. We assume that left vertices are strategic players, which could potentially manipulate the reporting of their true cost, and hence seek a truthful algorithm, i.e., no incoming vertex has incentive to misreport its bid. 
We continue to work under the secretarial model of input and the  large market assumption (Assumption \ref{ass:1}).


\begin{rem}As shown in \cite{yang2012crowdsourcing}, if bids of left vertices are used as payments, there is  incentive for left vertices to misreport their bids, and consequently the mechanism is not truthful or incentive compatible. Thus, the payment strategy is non-trivial. 
\end{rem} 



\begin{assumption}\label{ass:2} In the secretarial (uniformly random) left vertex arrival model, we also assume that for two different edges $e_1$ and $e_2$ with distinct left vertices $\ell_1$ and $\ell_2$ arriving at locations $\pi(1)$ and $\pi(2)$, if $v(e_1) > v(e_2)$, then $P(c(\ell_1) < c(\ell_2)) = 1/2$.
\end{assumption}

\begin{algorithm}
\caption{\textsc{ON-truth} Algorithm}\label{alg:msandp}
\begin{algorithmic}[1]
\State {\bf Input:} $L$ set of left vertices/users that arrive sequentially with permutation $\pi$, $R$ set of right vertices, Payment Budget $C$  
\State \%Offline Phase
\State $p = \frac{1}{2} $, $k \leftarrow Binomial(|L|, p)$
\State Let $L'$ be the first $k$ vertices of $L$
\State Run \textsc{\textsc{Threshold}} on $G'= (L' \cup R, E')$ to obtain $\gamma' \triangleq\gamma_C(G')$ and matching $\sfM_1$
\For{each right vertex $r: e = (\ell,r) \in \sfM_1$}
	\State Set $\text{reward}(r):=v(e)$ and $\text{cost}(r):=c(\ell)$ 
	\EndFor
	\For{each right vertex $r: (*,r) \notin \sfM_1$}
	\State Set $\text{reward}(r):=0$ and $\text{cost}(r):=0$
	\EndFor
	\State \%Decision Phase
\State$\sfM_{\mathsf{ON-T}} =\emptyset$ 
\For{every new left vertex $\ell \in L \backslash L'$},
	\State \%Pruning: Delete all edges $e = (\ell, r), r\in R$ s.t. $b(e)>\gamma'$
	\State Let $e^\star = (\ell, r)$ be the edge with the largest value such that $v(e^\star) \ge \text{reward}(r)\ 
	 \textbf{AND} \ c(\ell) \le \text{cost}(r) $	
	\If{$\sfM_{\mathsf{ON-T}} \cup  \{e^{\star}\}$ is a matching  }
		\State $\sfM_{\mathsf{ON-T}} = \sfM_{\mathsf{ON-T}}  \cup \{e^\star\}$ 
	\State Pay $p_\ell = \gamma' v(e^\star)$ to vertex $\ell$
	\Else 
	\State Let $\ell$ be permanently unmatched
	\EndIf
\EndFor
\end{algorithmic}
\end{algorithm}
To solve the online truthful budgeted matching problem we propose the 
 \textsc{ON-truth} algorithm that is almost identical to the $\mathsf{ON}$ algorithm in terms of when a left vertex is selected.  The first  difference is in size $t$ of the set of left vertices over which the offline algorithm \textsc{threshold} is run. With $\mathsf{ON}$, $t=n/e$, while with \textsc{ON-truth}, $t= Binomial(n,1/2)$.  The second difference in setting the reward for a right vertex that is part of the offline matching to be equal to the value of the matched edge, instead of the buck-per-bang as in $\mathsf{ON}$.
 A new feature with \textsc{ON-truth} is the payment rule for any selected left vertex, and the payment for left vertex $\ell$ of the selected edge $e^\star$ is $\gamma_C(G')v(e^\star)$.

We first compute the expected utility of matching $\sfM_{\mathsf{ON-T}}$ produced by algorithm \textsc{ON-truth} without enforcing the condition $c(\ell) \le \text{cost}(r)$ for selecting a left vertex on Line 17,  where the expectation is over the uniformly random left vertex arrival sequences.

\begin{algorithm}
\caption{$\mathsf{SIMULATE}$ Algorithm}\label{alg:simulate}
\begin{algorithmic}[1]
\State {\bf Input:} Graph $G$ and threshold $(\gamma)$ 
\State {\bf Output:} Matching $\sfM_{1s}, \sfM_{2s}$
\State Remove edges of $G$ with $b(e) > \gamma$ to get $G(\gamma)$
\State Sort edges of $G(\gamma)$ in decreasing order of their value
\State $\sfM_{1s} = \Phi, \sfM_{2s} = \Phi$
\State Mark each left vertex $\ell \in G(\gamma)$ as unassigned
\State For each edge $e = (\ell, r)$ in sorted order 
\If{$\ell$ is unassigned {\bf AND} $\sfM_{1s} \cup e$ is a matching} 
\State Mark $\ell$ as assigend 
\State Flip a coin with probability $p$ of heads
\State If heads, $\sfM_{1s} \leftarrow \sfM_{1s} \cup e$
\State else $\sfM_{2s} \leftarrow \sfM_{2s} \cup e$
\EndIf
\end{algorithmic}
\end{algorithm}
\vspace{-0.2in}
\begin{algorithm}
\caption{$\mathsf{SampleAndPermute}$ Algorithm}\label{alg:sp}
\begin{algorithmic}[1]
\State {\bf Input:} Graph $G = (L \cup R, E)$ 
\State {\bf Output:} Matching $\sfM_{2p}, \sfM_{3p}$
\State \%Offline Phase
\State $L' = \Phi$
\For{each $\ell \in L$}
\State With probability $\frac{1}{2}$, $L' \leftarrow L' \cup \ell$
\EndFor
\State $(\sfM_{1p}, \gamma') \rightarrow \textsc{Threshold}(G(L' \cup R, E(L') ))$
\For{each $r \in R$}
\State Set $\text{reward}(r)= v(e)$ if $e = (\ell, r) \in \sfM_{1p}$
\State Set $\text{reward}(r)= 0$ if $e = (*, r) \notin \sfM_{1p}$
\EndFor
\State \%Decision Phase
\State $\sfM_{2p} = \Phi, \sfM_{3p}= \Phi$
\For{each $\ell \in L \backslash L'$ and $b(e) \le \gamma'$} in random order
\State Let $e = (\ell, r)$ be the edge with largest value such that $v(e) \ge \text{reward}(r)$ 
\State Add $e$ to $\sfM_{2p}$.
\State If $\sfM_{3p} \cup e$ is a matching $\sfM_{3p} \leftarrow \sfM_{3p} \cup e$
\EndFor
\end{algorithmic}
\end{algorithm}

\begin{lemma}\label{lem:offguarantee} $\bbE\{v(\sfM_{\mathsf{ON-T}})\} \ge v(\mathsf{OPT}(C))/12$, when  condition $c(\ell) \le \text{cost}(r)$ is not enforced for selecting a left vertex in \textsc{ON-truth}.
\end{lemma}

To prove the result, we work with two intermediate algorithms \textsc{Simulate} and \textsc{SampleandPermute}, that will help in lower bounding the utility of the matching $\sfM_{\mathsf{ON-T}}$ produced by \textsc{On-truth}, similar to \cite{KorulaPal}. 
The connection between \textsc{SampleandPermute} and the proposed algorithm \textsc{On-truth}, is that the output matching $\sfM_{3p}$ of \textsc{Sampleandpermute} and $\sfM_{\mathsf{ON-T}}$ produced by \textsc{On-truth} are almost identical, except for the  difference in defining the set $L'$ (set of left vertices used to generate the threshold $\gamma'$), without enforcing condition  $c(\ell) \le \text{cost}(r)$ in \textsc{On-truth}. But with both these definitions, a left vertex is selected to be part of $L'$ with probability $1/2$ independently. Consequently, the  utilities of matchings $\sfM_{3p}$ and $\sfM_{\mathsf{ON-T}}$ are identical in expectation.
So to lower bound the utility of $\sfM_{\mathsf{ON-T}}$, we find a lower bound on the utility of $\sfM_{3p}$ of \textsc{Sampleandpermute}, and focus entirely on \textsc{Sampleandpermute} algorithm as follows.

\textsc{Simulate} is an offline matching algorithm, where each edge in descending order of its value is either assigned to matching $\sfM_{1s}$ or pseudo matching $\sfM_{2s}$\footnote{$\sfM_{2s}$ is not a matching since in $\sfM_{2s}$ multiple edges can be incident on any right vertex.} depending on the coin toss for that edge. Important to notice is that for \textsc{Simulate}, once a coin is tossed for an edge making the left vertex assigned, no other coin is tossed for any edge that shares a common left vertex with it. So it is essentially identical to tossing a coin once for each left vertex instead of each individual edge as done in algorithm \textsc{Sampleandpermute}. Thus, whenever coin tosses are identical for  \textsc{Sampleandpermute} and \textsc{Simulate}, and the 
$\gamma'$ computed by \textsc{threshold} algorithm invoked inside \textsc{Sampleandpermute} is identical to the input $\gamma$ to \textsc{Simulate},  it is easy to see that the matching $\sfM_{1s} = \sfM_{1p} $ and pseudo matching $\sfM_{2s} = \sfM_{2p}$ produced by \textsc{Simulate} and \textsc{Sampleandpermute} \cite{KorulaPal}. 

\begin{lemma}[Lemma 2.3 \cite{KorulaPal}]\label{lem:KP} For \textsc{Simulate} algorithm, if the input threshold $\gamma$ and the coin tosses for choosing an edge (to be part of $\sfM_{1s}$ or $\sfM_{2s}$) are independent, then $\bbE\{v(\sfM_{2s})\} = 
\bbE\{v(\sfM_{1s})\}$.
\end{lemma}

\begin{rem} Lemma \ref{lem:offguarantee} would be directly provable following the techniques of \cite{KorulaPal}, if Lemma \ref{lem:KP} could be applied on the matching $\sfM_{1s}$ and $\sfM_{2s}$, for the case when $\sfM_{1s} = \sfM_{1p} $ and pseudo matchings $\sfM_{2s} = \sfM_{2p}$.  Problem is that $\sfM_{1s} = \sfM_{1p} $ and pseudo matchings $\sfM_{2s} = \sfM_{2p}$ only when the respective coin tosses in \textsc{Simulate} and \textsc{Sampleandpermute}, and the
$\gamma$ (input to \textsc{Simulate}) and  $\gamma'$ (computed by \textsc{threshold} algorithm invoked inside \textsc{Sampleandpermute}) are identical. Since $\gamma'$ is dependent on coin tosses of \textsc{Sampleandpermute}, so if $\gamma'$ in input to \textsc{Simulate} and the coin tosses are identical to as in \textsc{Sampleandpermute}, they are dependent on each other, and Lemma \ref{lem:KP} is not applicable.
\end{rem}

So the proof of Lemma \ref{lem:offguarantee} is more involved as presented next.

\begin{proof} 
Consider the full graph $G = (L \cup R, E)$ (offline),  and let $\gamma_f$ be the output threshold when   \textsc{Threshold} is run over the full graph $G$.  

Toss $2|L|$ coins independently with heads probability $1/2$, and record their outcomes in two vectors $\bt_1 = [t_{11} \dots t_{1|L|}]$ and $\bt_2 = [t_{21} \dots t_{2|L|}]$, where $t_{ij} = 1$ if the $(i,j)^{th}$ coin toss is heads, and $0$ otherwise.  We will associate $\bt_1$ with coin tosses for the $|L|$ left vertices while running \textsc{simulate} with full graph $G$ and threshold $\gamma_f$.

All left vertices for which $t_{2j} = 1$ (set $L'$ as defined in \textsc{SampleAndPermute}), will be part of the offline phase and the remaining left vertices with $t_{2j} = 0$ will be part of decision/online phase in algorithm  \textsc{SampleAndPermute}.
Let $\gamma'_{\bt_2}$ be the output threshold from $\textsc{Threshold}$ when executed inside the algorithm in \textsc{SampleAndPermute} with coin toss vector $\bt_2$, which is identical to running  algorithm $\textsc{Threshold}$ on graph
$(G(L' \cup R, E'))$. 
\begin{rem} Note that for fixed coin tosses $\bt_2$ that determines $\gamma'_{\bt_2}$ completely, the matchings 
$\sfM_{1p}, \sfM_{2p}$ produced by \textsc{SampleAndPermute} are identical to the matchings $\sfM_{1s}, \sfM_{2s}$ produced by \textsc{Simulate} with input $G(L\cup R, E)$ and threshold $\gamma'_{\bt_2}$, and coin tosses $\bt_2$, respectively. Hence \begin{equation}\label{eq:dum90}
\bbE\{v(\sfM_{2p}(G(\gamma_{\bt})))\}=\bbE\{v(\sfM_{2s}(G(\gamma_{\bt})))\}\footnote{$\sfM(G(\gamma))$ is the matching obtained with graph $G$ and threshold $\gamma$.}
\end{equation}
\end{rem}

For fixed realizations of coin tosses $\bt_1$ and $\bt_2$, now we compare the pseudo matchings $\sfM_{2s}$ produced by algorithm \textsc{simulate} with input graph $G = (L \cup R, E)$, threshold $\gamma_f$ with coin tosses $\bt_1$, and input graph $G(L' \cup R, E')$ and threshold $\gamma_{\bt_2}$ with coin tosses $\bt_2$, respectively. 

Let the set of left vertices that have at least one edge in $G(\gamma_f)$ and in $G(\gamma_{\bt_2})$ be $L_1$ and $L_2$, respectively. Recall that $G(\gamma)$ contains all edges $e \in G$ that have $b(e) \le \gamma$. From Lemma \ref{lem:monotonegreedymatching}, we know that any choice of $\bt_2$, $\gamma_{\bt_2} \ge \gamma_f$. Hence $L_1 \subseteq L_2$. 
Consider the case when the realization of coin tosses $\bt_1$ and $\bt_2$ restricted to set $L_1$ of left vertices  be the same.
Then independent of the coin tosses for vertices $L_2 \backslash L_1$, we have that for algorithm \textsc{Simulate}
\begin{equation}\label{eq:dum1}
v(\sfM_{2s}(G(\gamma_{\bt_2}))) \ge v(\sfM_{2s}(G(\gamma_f))),
\end{equation} since $\gamma_{\bt} \ge \gamma_f$ and each edge present in $G(\gamma_f)$ is also present in $G(\gamma_{\bt_2})$ and $\sfM_{2s}$ is a pseudo matching and accepts multiple edges incident on any right vertex. For pseudo matching we let $v(\sfM_{2s}(G(\gamma))) = \sum_{e \in \sfM_{2s}(G(\gamma))} v(e)$. Thus, taking expectation of \eqref{eq:dum1}, we have that 
\begin{equation}\label{eq:dum2}
\bbE\{v(\sfM_{2s}(G(\gamma_{\bt_2})))\} \ge \bbE\left\{v(\sfM_{2s}(G(\gamma_f)))\right\}.
\end{equation}

Since $\gamma_f$ (obtained by running \textsc{Threshold} on full graph $G$) does not depend on any coin tosses $\bt_1$ or $\bt_2$, we have from Lemma \ref{lem:KP}, 
\begin{equation}\label{eq:dum3}\bbE\{v(\sfM_{2s}(G(\gamma_f)))\} = 
\bbE\{v(\sfM_{1s}(G(\gamma_f)))\},\end{equation} where the expectation is over the coin tosses $\bt_1$.

Next, we lower bound the $\bbE\{v(\sfM_{1s}(G(\gamma_f)))\}$.
Consider graph $G$ and threshold $\gamma_f$ as an input to the \textsc{simulate} algorithm. For a fixed realization of $\bt_1$, let $\mathsf{OPT}_{1/2}(\bt_1)$ be the optimal matching considering only left vertices $j$ for which coin tosses $t_{1j}=1$. Taking the expectation with respect to $\bt_1$,  we have that 
\begin{equation}\label{eq:dum10}
\bbE\{v(\mathsf{OPT}_{1/2})\} = \frac{\bbE\{v(\mathsf{OPT})\}}{2}.
\end{equation}

Moreover, as pointed out earlier in Remark \ref{rem:threshold}, the matching produced by \textsc{threshold} with output threshold $\hat {\gamma}$ is equivalent to finding a greedy matching with graph $G(\hat {\gamma})$. The same is true for matching $\sfM_{1s}$ produced by \textsc{simulate}. Thus, considering \textsc{simulate} algorithm with input $G$ and threshold $\gamma$, and all left vertices with $t_{1j}=1$ with $\bt_1$,  
 from Lemma  
 \ref{lem:umguarantee}, 
 $v(\sfM_{1s}(G(\gamma)) \ge \frac{v(\mathsf{OPT}_{1/2}(\bt_1))}{3}$. 
Thus,  taking the expectation with respect to $\bt_1$, from \eqref{eq:dum10}, we get 
\begin{equation}\label{eq:dum4}
\bbE\{v(\sfM_{1s}(G(\gamma)))\} = \frac{v(\mathsf{OPT}(C))}{6}.
\end{equation}

From Lemma 2.5  \cite{KorulaPal}, we have that for the \textsc{SampleandPermute} algorithm, 
\begin{equation}\label{eq:dum5}\bbE\{v(\sfM_{3p}(G(\gamma_{\bt})))\} \ge \frac{\bbE\{v(\sfM_{2p}(G(\gamma_{\bt})))\}}{2} ,\end{equation} since the pruning step to obtain $\sfM_{3p}$ from $\sfM_{2p}$ in \textsc{SampleandPermute} algorithm only depends on the relative order in which vertices (with coin tosses $t_{2j} =0$) arrive in the decision phase, and not on coin tosses themselves. 


Combining  \eqref{eq:dum2}, \eqref{eq:dum3}, \eqref{eq:dum4}, \eqref{eq:dum5}, and \eqref{eq:dum90}
 we get that 
 $$\bbE\{v(\sfM_3(\gamma_{\bt}))\} \ge \frac{v(\mathsf{OPT}(C))}{12}.$$
Since the expected utility of matching $\sfM_3(\gamma_{\bt_2})$ of \textsc{Sampleandpermute} is same as the expected utility of matching $\sfM_{\textsf{ON-T}}$ (the output of \textsc{ON-truth}), we have the result.
\end{proof}

The following theorem is the second main result of the paper.
\begin{theorem}
Algorithm \textsc{ON-truth} is $24$-competitive, satisfies the payment budget constraint, payment is always larger than the bid  for each selected left vertex, i.e., $p_{\ell}\ge c(\ell)$, and is truthful.\end{theorem}
\begin{proof} 
Disregarding the condition $c(\ell) < \text{cost}(r)$, an edge $e$ incident on a left vertex $\ell$ is chosen by algorithm \textsc{ON-truth} if its value $v(e)$ is larger than the reward (value of an edge of a right vertex that is part of the offline matching). Hence from Assumption \ref{ass:2}, if a left vertex is accepted without the condition $c(\ell) < \text{cost}(r)$, then it is still accepted with probability $\frac{1}{2}$ while enforcing the condition $c(\ell) < \text{cost}(r)$.
Combining this fact with Lemma  \ref{lem:offguarantee}, we get the $24 $-competitiveness of \textsc{ON-truth} algorithm.

The claim that $p_{\ell} \ge c(\ell)$ for each selected left vertex $\ell$, follows from the fact that each left vertex is considered in the decision phase only if its buck per bang $b(e) = \frac{c(\ell)}{v(\ell)} < \gamma$. Since $p_{\ell} = \gamma v(\ell)$, clearly, $p_{\ell}\ge c(\ell)$.
The budget feasibility and incentive compatibility are shown in Lemma \ref{lem:budfeasExp} and \ref{lem:icExp}, respectively.
\end{proof}
\begin{lemma}\label{lem:budfeasExp}
Algorithm \textsc{ON-truth} satisfies the payment budget constraint.
\end{lemma}
\begin{proof} Similar to Lemma \ref{lem:budfeas}, by enforcing the condition that any left vertex is selected only if $c(\ell) \le \text{cost}(r)$ for some unmatched right vertex that is part of $\sfM_1$, and 
$\sum_{r: e=(\ell, r) \in \sfM_1} \text{cost}(r) \le C$ for the \textsc{Threshold} algorithm by Remark \ref{rem:threshold}.
\end{proof}

Next, we show the most important property of \textsc{ON-truth}, its truthfulness. Towards that end, we will use the Myerson's Theorem \cite{myerson1981optimal}.
\begin{theorem}\cite{myerson1981optimal}\label{Myerson_Theorem}
A reverse auction is truthful if and only if:
\begin{itemize}
\item The selection rule is monotone. If a user $\ell$ wins the auction by bidding $c(\ell)$, it would also win the auction by bidding an amount $c(\ell)'$, where $c(\ell)' < c(\ell)$.
\item Each winner is paid a critical amount. If a winning user submits a bid greater than this critical value, it will not get selected.
\end{itemize}
\end{theorem} 
\begin{lemma}\label{lem:icExp}
\textsc{ON-truth} is a truthful online algorithm.
\end{lemma}
\begin{proof} We show that the two conditions of Theorem \ref{Myerson_Theorem} are satisfied for the \textsc{ON-truth} algorithm, similar to \cite{VazeMatching2016}.  In the decision phase, if any left vertex reduces its bid, then clearly its buck per bang $b(e)$ decreases, and hence it is still accepted if it was accepted before. Thus, monotone condition is satisfied.

The criticality of payment is shown as follows. Note that the payment made by \textsc{ON-truth} to a selected left vertex $\ell$ is $p_{\ell} = \gamma' v(e), e =(\ell,r) \in \sfM_{\mathsf{ON-T}}$, where the right vertex index $r$ is such that utility $v(e), e=(\ell, r)$ is largest among the unmatched right vertices at the time of arrival of vertex $\ell$ that have an edge to left vertex  $\ell$, and $v(e) > \text{reward}(r)$. 

Now, if suppose the bid $c(\ell)$ of left vertex $\ell$ is more than $p_{\ell}=\gamma' v(e)$, then its bang per  buck $c(\ell)/v(e) >\gamma'$. Moreover, since $v(e) > v(e')$ for all edges $e'$ incident on unmatched right vertices from $\ell$ at the arrival of left vertex $\ell$, we have that $c(\ell)/v(e') >  \gamma'$. 
Thus, all edges out of left vertex $\ell$ incident on currently unmatched right vertices are removed in the pruning stage of the decision phase, and hence vertex $\ell$ cannot be selected.
\end{proof}

{\it Discussion:} In this section, we proposed a $24$-competitive online algorithm for the budgeted bipartite matching problem that is also truthful, improving upon the best known bound of $24 \beta$-competitive \cite{VazeMatching2016}, where $\beta$ is the ratio of the maximum and minimum utility of any edge. The proposed algorithm is a significant/fundamental improvement over prior work, since it eliminates any dependence on the system/input parameters, making it scalable and suitable for large networks.

%
%
\section{Simulation }
We consider the uplink of a single cell of cellular communication for the D2D application, where $150$ cellulars users are present with one basestation. Out of $150$ nodes, the helper set is of size $n=50$, while the rest $100$ nodes (set $R$) are seeking help. The payment budget constraint is $100$. All users are assumed to be uniformly located in the coverage area, and the utility between any helper and a  help seeking node is drawn uniformly from $[0,20]$, and the bid for each helper is drawn uniformly from $[0,5]$. Let $\delta$ be the fraction of nodes any one helper can help, and we assume that for fixed $\delta$, the nodes that any helper can help are uniformly distributed among the $100$ nodes. In Fig. \ref{fig:cr}, we plot the competitive ratio of the  proposed algorithm \textsc{On-truth} as a function of $\delta$. We see that the competitive ratio of \textsc{On-truth}  algorithm is far better than the derived guarantee ($24$-competitive). An important observation 
from Fig. \ref{fig:cr} is that as $\delta$ increases, the competitive ratio increases significantly, since with larger $\delta$, the quality of the offline matching and the number of right vertices matched in the offline phase increases, allowing the \textsc{On-truth} algorithm to match larger number of left vertices, and extract larger utility.

%

\begin{figure} \label{fig:cr}
\begin{tikzpicture}[xscale=.85, yscale=.45] 
\draw[help lines] (0,0) grid (9,10);,
\draw [ultra thick, dashed](0,0) node [left] {$0$} (0,1) node [left] {$.03$} (0,2) node [left]  {$.06$} (0,3) node [left]  {$.09$} 
(0,4) node [left]  {$.12$}  (0,5) node [left]  {$.15$} (0,6) node [left]  {$.18$} (0,7) node [left]  {$.21$} (0,8) node [left]  {$.24$} (0,9) node [left]  {$.27$} (0,1) node [left]  {$1$};
\draw (4.5,-1) node [below] {$\text{Mean out-degree of incoming left vertex} \ \delta$};

\draw (-1.5,5) node [below, rotate=90] {$\text{Competitive Ratio}$};

\draw [ultra thick, dashed]  (1,0) node [below]  {$.2$} (2,0) node [below]  {$.3$} 
(3,0) node [below]  {$.4$}  (4,0) node [below]  {$.5$} (5,0) node [below]  {$.6$} (6,0) node [below]  {$.7$} (7,0) node [below]  {$.8$} (8,0) node [below]  {$.9$};

\draw [blue, line width=6]
(1,0) -- (1, 4) (2,0) -- (2, 5.2) (3,0) -- (3, 6.1) (3,0) -- (3, 6.8) (4,0) -- (4, 7.5) (5,0) -- (5, 8.01) (6,0) -- (6, 8.7) (7,0) -- (7, 8.9) (8,0) -- (8, 9); 

\end{tikzpicture}
\caption{Competitive ratio of the \textsc{On-truth} algorithm as a function of $\delta$.}
\end{figure}
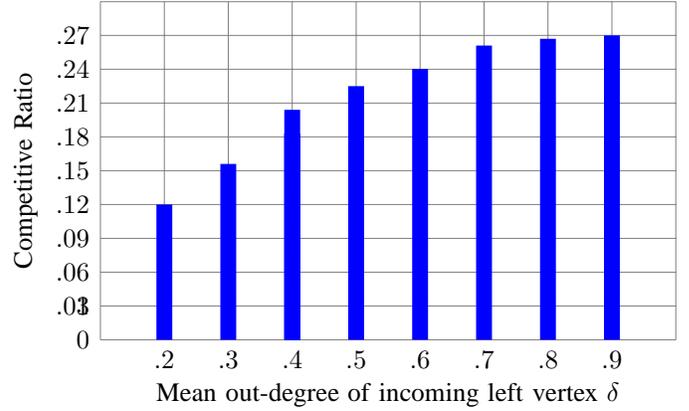

\section{Conclusions}
In this paper, we have made significant progress in finding  better online algorithms for bipartite matching under the capacity constraint on the 'size' of selected left vertices. Under the large market assumption, that is reasonable in practice, we are able to improve the best known competitive ratio for the knapsack problem from $10e$ to $2e$,  and from non-constant to constant for the truthful matching problem. 

%
%
\appendices

\section{Proof of Lemma \ref{lem:umguarantee}}
\begin{proof} 
Decompose the optimal fractional matching solution $\mathsf{OPT} = \{\mathsf{OPT}^{+} \cup \mathsf{OPT}^{-}\}$, where $\mathsf{OPT}^{+}$ contains edges of $\mathsf{OPT}$ that have $b(e) > \gamma_C$, and $\mathsf{OPT}^{-}$ contains edges of $\mathsf{OPT}$ that have $b(e) \le \gamma_C$. Similarly, let $\mathsf{OPT}(\gamma_C)$ be the optimal fractional matching on subgraph $G(\gamma_C)\subseteq G$, where $\gamma_C$ is the output threshold from the \textsc{Threshold} algorithm with graph $G$.  By definition of optimal matching, $v(\mathsf{OPT}^{-}) \le v(\mathsf{OPT}(\gamma_C))$. Moreover, for $\sfM$, the output matching from \textsc{Threshold} algorithm with graph $G$, we have 
$v(\sfM) \ge \frac{v(\mathsf{OPT}(\gamma_C))}{2}$, since $\sfM$ is a greedy matching on $G(\gamma)$ (subgraph with all edges having $b(e) \le \gamma_C$). Therefore, 
$v(\sfM) \ge \frac{v(\mathsf{OPT}^{-})}{2}$.

All edges $e = (\ell ,r) \in\mathsf{OPT}^{+}$, have $b(e) = \frac{c(\ell)}{v(e)} > \gamma_C$. Thus, 
$v(\mathsf{OPT}^{+}) = \sum_{e =(\ell, r)   \in \mathsf{OPT}^{+}} x(\ell) v(e) < \frac{\sum_{e = (\ell ,r) \in \mathsf{OPT}^{+}}x(\ell) c(\ell)}{\gamma_C}$, where $x(\ell)$ are fractional weights in the optimal solution. 
Moreover, the total budget constraint of $C$ ($\sum_{e = (\ell ,r) \in \mathsf{OPT}} x(\ell)c(\ell) \le C$) implies that $v(\mathsf{OPT}^{+}) < \frac{C}{\gamma_C}$. Assuming that the budget constraint is tight with the \textsc{Threshold} algorithm ($\sum_{e\in \sfM}\gamma_C v(e)= C$), $v(\sfM) = \frac{c}{\gamma_C}$. Therefore, 
$v(\mathsf{OPT}^{+}) < v(\sfM)$. Combining this with $v(\sfM) \ge \frac{v(\mathsf{OPT}^{-})}{2}$, we have $v(\sfM) \le 3 v(\mathsf{OPT})$ as required.

If the capacity constraint is not tight with the \textsc{Threshold} algorithm, then under Assumption \ref{ass:1} ($v_{max}/v(\mathsf{OPT})=o(1)$), by the definition of \textsc{Threshold} algorithm that finds the largest feasible $\gamma$, the leftover capacity $C-\sum_{e\in \sfM}\gamma_C v(e)$ is no more than $\gamma_C v_{max}$, and similar argument gives us that $v(\mathsf{OPT}) \le (3+o(1)) v(\sfM)$. 
\end{proof}

\section{Proof of Corollary \ref{cor:umguarantee}}
\begin{proof} 
Compared to the proof of Lemma \ref{lem:umguarantee}, the only difference is in contribution from edges with $b(e) \le \gamma_C$. Since all edge weights incident on any left vertex are identical and the number of right vertices is equal to left vertices, greedy matching $\sfM$ is actually optimal for edges with $b(e) \le \gamma_C$. Hence $v(\sfM) =  v(\mathsf{OPT}(\gamma_C))$. 
Noting that  $v(\mathsf{OPT}(\gamma_C)) \ge
v(\mathsf{OPT}^{-})$, we get 
$v(\sfM) \ge
v(\mathsf{OPT}^{-})$, and from which the result follows, since $v(\sfM) \ge
v(\mathsf{OPT}^{+})$ ( proof of Lemma \ref{lem:umguarantee}).
\end{proof}

\section{Proof of Lemma \ref{lem:monotonegreedymatching}}
\begin{proof} 
When a left vertex is removed (by deleting all edges incident to it as considered), the proof of claim $1$ follows by standard procedure by considering each right vertex, for which the value of the matched edge in $\text{\textsc{Greedy}}(G)$ is at least as much as in  $\text{\textsc{Greedy}}(F)$. 

For the second  and third claim, note that an edge $e$ incident on left vertex $\ell$ is removed in $G(\gamma)$ compared to $G$, if 
$b(e) > \gamma$ or equivalently if $v(e) < \frac{c(\ell)}{\gamma}$. 
Recall that the cost of any edge only depends on the index of its left vertex. Hence, if edge $e =(\ell, r)$ is removed from $G$ to obtain $G(\gamma)$, 
then all the edges $e'$ incident on $\ell$ with utility $v(e') < v(e) $ are also removed. So essentially, edges are removed monotonically from $G$ to produce $G(\gamma)$. So the proofs for the second and third claim follow similarly to the first.
\end{proof}

\section{Proof of Lemma \ref{lem:polytimecomplexity}}
\begin{proof} From the definition of  Algorithm \textsc{Threshold} its clear that if any $\gamma \in \cA(G)$, then $\gamma_C \ge \gamma$. Hence the key step is to show that if any $\gamma \notin \cA(G)$, then $\gamma_C < \gamma$ which follows from the second claim of Lemma \ref{lem:monotonegreedymatching}, that $v(\textsc{Greedy}(G(\gamma_1))) \ge v(\textsc{Greedy}(G(\gamma_2)))$ for $\gamma_1 \ge \gamma_2$. Therefore, if for any $\gamma \notin \cA(G)$, then for any $\gamma' > \gamma$, $\gamma' \notin \cA(G)$. Hence we can use bisection to find the maximum.
\end{proof}

\section{Proof of Lemma \ref{lem:monotoneGamma}}
\begin{proof}
From Lemma \ref{lem:monotonegreedymatching}, \begin{equation}\label{eq:FgGg}
v(\sfM(F(\gamma))) \le v(\text{\textsc{Greedy}}(G(\gamma))).
\end{equation}
Let the threshold and the matching obtained by running \textsc{Threshold} on $G$ with budget $C$ be $\gamma_C(G) = \gamma$, and $\sfM(G)$, respectively, where $\gamma \le \frac{C}{v(\text{\textsc{Greedy}}(G(\gamma)))}$. Now we consider $F(\gamma)$ as the input graph to the \textsc{Threshold} with same budget constraint $C$. Since $\gamma \le \frac{C}{v(\text{\textsc{Greedy}}(G(\gamma)))}$, from \eqref{eq:FgGg}, clearly, 
$\gamma \le \frac{C}{v(\sfM(F(\gamma))}$, and $\sum_{e\in \sfM(F(\gamma))}\gamma v(e)\leq C$. Therefore,  $\gamma \in \cA(F)$, which by definition of $\gamma_C(F)$ implies $\gamma_C(F) \ge \gamma$.

\end{proof}
\bibliographystyle{IEEEtran}
\bibliography{onlined2d}


\end{document}